\documentclass[onecolumn, draftclsnofoot]{IEEEtran}

\usepackage{amsmath, amsfonts, amsthm}
\usepackage{xcolor}
\usepackage{colortbl}
\usepackage{epsfig}
\usepackage{comment}
\usepackage{dsfont}
\usepackage{setspace}
\usepackage{microtype}
\usepackage{bm}

\usepackage{array}   
\newcolumntype{L}{>{$}l<{$}} 
\newcolumntype{C}{>{$}c<{$}} 
\newcolumntype{R}{>{$}r<{$}} 

\DeclareMathOperator*{\argmin}{arg\, min}
\newcommand{\F}{\mathbb{F}}
\newcommand{\N}{\mathbb{N}}
\newcommand{\R}{\mathbb{R}}
\newcommand{\C}{\mathbb{C}}
\renewcommand{\tilde}{\widetilde}

\newenvironment{psmallmatrix}
  {\left(\begin{smallmatrix}}
  {\end{smallmatrix}\right)}

\theoremstyle{plain}
\newtheorem{theorem}{Theorem}

\theoremstyle{definition}
\newtheorem*{example}{Example}

\theoremstyle{remark}
\newtheorem*{remark}{Remark}

\title{Secure Distributed Gram Matrix Multiplication}
\author{
\IEEEauthorblockN{Okko Makkonen,~\IEEEmembership{Graduate Student Member,~IEEE} and Camilla Hollanti,~\IEEEmembership{Member,~IEEE}}

\IEEEauthorblockA{
    Department of Mathematics and Systems Analysis \\
    Aalto University, Finland \\
    \texttt{\{okko.makkonen, camilla.hollanti\}@aalto.fi}
}
}
\IEEEoverridecommandlockouts

\begin{document}

\maketitle

\begin{abstract}
The \emph{Gram matrix} of a matrix $A$ is defined as $AA^T$ (or $A^T\!A$). Computing the Gram matrix is an important operation in many applications, such as linear regression with the least squares method, where the explicit solution formula includes the Gram matrix of the data matrix. Secure distributed matrix multiplication (SDMM) can be used to compute the product of two matrices using the help of worker servers. If a Gram matrix were computed using SDMM, the data matrix would need to be encoded twice, which causes an unnecessary overhead in the communication cost. We propose a new scheme for this purpose called secure distributed Gram matrix multiplication (SDGMM). It can leverage the advantages of computing a Gram matrix instead of a regular matrix product.
\end{abstract}

\section{Introduction}

Secure distributed matrix multiplication (SDMM) has been used to compute the product of two matrices using the help of worker servers, while keeping the matrices information-theoretically secure against colluding workers. SDMM was first introduced by Chang and Tandon in \cite{chang2018capacity}. Many different schemes have since been presented, such as the secure MatDot scheme in \cite{aliasgari2019distributed}, the GASP scheme in \cite{d2020gasp} and the DFT scheme in \cite{mital2022secure}. These schemes compute the products over finite fields, but computations over complex numbers have also been considered in \cite{makkonen2022analog, soleymani2020privacy, soleymani2021analog}. Other secure coded computation tasks have also been considered, such as batch matrix multiplication in \cite{jia2021capacity, li2022private}, arbitrary polynomial evaluation using Lagrange coded computation in \cite{yu2019lagrange}, and convolution in \cite{yang2019secure}.

In this paper we will consider the computation of the Gram matrix $AA^T$ for some matrix $A$. In general, the Gram matrix is defined using some inner product. We want to consider this multiplication also over finite fields, where the concept of inner product is not meaningful. Hence, we only consider the product $AA^T$. Such a computation is an important step in many applications. For example, if $A$ is a $t \times s$ matrix representing centered data of $t$ variables and $s$ observations, then an estimate of the covariance matrix can be computed by $\frac{1}{s - 1} AA^T$. Another example is linear regression using the least squares method, where the aim is to solve the following minimization problem
\begin{equation*}
    \hat{\beta} = \argmin_{\beta \in \R^t} \lVert A^T\beta - b \rVert_2
\end{equation*}
for some given data matrix $A \in \R^{t \times s}$ and a vector $b \in \R^s$. If the rows of $A$ are linearly independent, then the solution obtains the explicit form
\begin{equation*}
    \hat{\beta} = (AA^T)^{-1}Ab.
\end{equation*}
Here, computing the Gram matrix of $A$ is only one step in the whole computation. However, if $A$ is a $t \times s$ matrix, where $t \ll s$, then inverting the $t \times t$ matrix $AA^T$ is less computationally heavy than computing the Gram matrix $AA^T$.

The Gram matrix of the matrix $A$ can be computed more efficiently compared to regular matrix multiplication of the matrices $A$ and $A^T$. The Gram matrix is symmetric, which means that only the lower (or upper) triangular part needs to be computed and stored as the other half can be obtained by symmetry. For more efficient algorithms on Gram matrix multiplication, see \cite{dumas2020fast}.

A Gram matrix can easily be computed using SDMM by considering the matrix multiplication of $A$ and $A^T$. However, this would not be particularly efficient, since it would involve encoding the matrix $A$ essentially twice. As communication cost is an important consideration, it is valuable to find even small decreases in the upload cost. Additionally, the multiplication performed by the workers can not utilize the benefits of Gram matrix multiplication, since the results will not be symmetric in general. This leads us to study some variants of standard SDMM schemes, which can be utilized to compute the Gram matrix efficiently. The idea is to encode the matrix $A$ just once, such that the computational task of the workers consists of computing the Gram matrix of the encoded pieces. This approach requires a nontrivial choice of the linear code, so that the product can be obtained from these partial results.

Our proposed methods will be based on degree tables, which were first introduced in \cite{d2021degree}. However, we consider symmetric degree tables, which have not been studied previously to our knowledge. The Gram matrix multiplication has also been considered in \cite{yu2019lagrange} with the Lagrange coded computation scheme, but the scheme was not optimized for this specific problem. We also present an analog version of the scheme, which has much lower communication cost, but is not able to provide security.

\section{Preliminaries}

We write $[n] = \{1, \dots, n\}$ and $\N = \{0, 1, 2, \dots\}$. Unless mentioned otherwise, we consider matrices and scalars over a finite field $\F_q$ with $q$ elements. The group of units of $\F_q$ is denoted by $\F_q^\times = \F_q \setminus \{0\}$. Vectors in $\F_q^n$ are considered to be row vectors. If $G$ is a matrix, then $G^{\leq m}$ and $G^{> m}$ denote the submatrices with the first $m$ rows and the the rest of the rows, respectively. Furthermore, if $\mathcal{I}$ is a set of indices, then $G_\mathcal{I}$ is the submatrix of $G$ with the columns indexed by $\mathcal{I}$.

The key tool in distributed coded computation is partitioning the task into smaller subtasks. For matrices this is achieved by partitioning the matrix into equal sized block matrices. One common partitioning is the \emph{inner product partitioning (IPP)}, where the matrices $A$ and $B$ are partitioned such that
\begin{equation*}
    A = \begin{pmatrix} A_1 & \dots&  A_p \end{pmatrix} \quad \text{and} \quad B = \begin{psmallmatrix} B_1 \\ \vdots \\ B_p \end{psmallmatrix}.
\end{equation*}
Then their product can be expressed as $AB = \sum_{j=1}^p A_jB_j$. Another partitioning is the \emph{outer product partitioning (OPP)}, where the matrices $A$ and $B$ are partitioned such that
\begin{equation*}
    A = \begin{psmallmatrix} A_1 \\ \vdots \\ A_m \end{psmallmatrix} \quad \text{and} \quad B = \begin{pmatrix} B_1 & \dots & B_n \end{pmatrix}.
\end{equation*}
Then their product can be expressed as
\begin{equation*}
    AB = \begin{psmallmatrix}
    A_1B_1 & \dots & A_1B_n \\
    \vdots & \ddots & \vdots \\
    A_mB_1 & \dots & A_mB_n
    \end{psmallmatrix}.
\end{equation*}

The aim of distributed matrix multiplication is to partition the computation into smaller subcomputations, which can be distributed to the worker servers. Secure distributed matrix multiplication puts an additional constraint, which states that no information about the original matrices should be leaked during this process. Most SDMM schemes in the literature follow some kind of linear structure, where the encoded matrices are obtained as linear combinations of the matrices and some random noise that is added. This structure has been presented as a general framework called linear SDMM in \cite{makkonen2022general}. We shall present our contributions using the help of this framework.

A linear SDMM model with $N$ workers and $X$ colluding workers follows the steps below.
\begin{itemize}
    \item Partitioning the matrices $A$ and $B$ into $m$ and $n$ equal sized pieces, respectively.
    \item Computing the encodings
    \begin{align*}
        (\tilde{A}_1, \dots, \tilde{A}_N) &= (A_1, \dots, A_m, R_1, \dots, R_X)F, \\
        (\tilde{B}_1, \dots, \tilde{B}_N) &= (B_1, \dots, B_n, S_1, \dots, S_X)G,
    \end{align*}
    where the matrices $R_1, \dots, R_X$ and $S_1, \dots, S_X$ are chosen uniformly at random. Here $F, G$ are generator matrices of some linear codes.
    \item Worker $i$ computes the product $\tilde{C}_i = \tilde{A}_i\tilde{B}_i$ and returns it to the user.
    \item The user recovers the product $AB$ by computing some linear combinations of the responses $\tilde{C}_i$.
\end{itemize}
The \emph{recovery threshold} $R$ of an SDMM scheme is the minimal number of responses the user needs to perform the decoding in the worst case. The difference $N - R$ describes the maximal number of unresponsive workers that can be tolerated by the scheme.

An SDMM scheme is secure against $X$-collusion if
\begin{equation*}
    I(\bm{A}, \bm{B}; \tilde{\bm{A}}_\mathcal{X}, \tilde{\bm{B}}_\mathcal{X}) = 0
\end{equation*}
for all $\mathcal{X} \subset [N]$, $\lvert \mathcal{X} \rvert = X$. In \cite{makkonen2022general} it was shown that a linear SDMM scheme is secure against $X$-collusion if $F^{>m}$ and $G^{>n}$ generate MDS codes.

One way to construct a linear SDMM scheme is by considering the generator matrices $F$ and $G$ as generalized Vandermonde matrices with distinct evaluation points $\alpha \in \F_q^N$, and exponents $\varphi \in \N^{m + X}$ and $\gamma \in \N^{n + X}$, respectively. This choice of generator matrices means that the encoded pieces are evaluations of the polynomials
\begin{align*}
    f(x) &= \sum_{j=1}^m A_j x^{\varphi_j} + \sum_{k=1}^X R_k x^{\varphi_{m + k}},~\text{and}\\
    g(x) &= \sum_{j'=1}^n B_{j'} x^{\gamma_{j'}} + \sum_{k'=1}^X S_{k'} x^{\gamma_{n + k'}}
\end{align*}
at the points $\alpha_1, \dots, \alpha_N$. Then the products $\tilde{C}_i = \tilde{A}_i\tilde{B}_i$ will be evaluations of
\begin{align*}
    h(x) &= f(x)g(x) \\
    &= \sum_{j=1}^m \sum_{j'=1}^n A_jB_{j'} x^{\varphi_j + \gamma_{j'}} + \sum_{j=1}^m \sum_{k'=1}^X A_jS_{k'} x^{\varphi_j + \gamma_{n + k'}} \\
    &+ \sum_{k=1}^X \sum_{j'=1}^n R_kB_{j'} x^{\varphi_{m + k} + \gamma_{j'}} + \sum_{k=1}^X \sum_{k'=1}^X R_kS_{k'} x^{\varphi_{m + k} + \gamma_{n + k'}}.
\end{align*}
The decoding is usually thought of as obtaining the coefficients of $h(x)$ from the evaluations $\tilde{C}_i$. This leads us to consider the exponents of $h(x)$, which are obtained as the sum of the $\varphi$'s and $\gamma$'s. The table that is obtained as the outer sum of $\varphi$ and $\gamma$ is called the \emph{degree table} (denoted by $\varphi \oplus \gamma$), which was first introduced in \cite{d2021degree}. The general form of a degree table can be found in Table \ref{tab:degree_table}.

\begin{table*}[t]
    \centering
    \begin{tabular}{L|LLL|LLL}
        & \gamma_1 & \dots & \gamma_n & \gamma_{n + 1} & \dots &  \gamma_{n + X} \\ \hline
        \varphi_1 & \varphi_1 + \gamma_1 & \dots & \varphi_1 + \gamma_n & \varphi_1 + \gamma_{n + 1} & \dots & \varphi_1 + \gamma_{n + X} \\
        \vdots & \vdots & \ddots & \vdots & \vdots & \ddots & \vdots \\
        \varphi_m & \varphi_m + \gamma_1 & \dots & \varphi_m + \gamma_n & \varphi_m + \gamma_{n + 1} & \dots & \varphi_m + \gamma_{n + X} \\ \hline
        \varphi_{m + 1} & \varphi_{m + 1} + \gamma_1 & \dots & \varphi_{m + 1} + \gamma_n & \varphi_{m + 1} + \gamma_{n + 1} & \dots & \varphi_{m + 1} + \gamma_{n + X} \\
        \vdots & \vdots & \ddots & \vdots & \vdots & \ddots & \vdots \\
        \varphi_{m + X} & \varphi_{m + X} + \gamma_1 & \dots & \varphi_{m + X} + \gamma_n & \varphi_{m + X} + \gamma_{n + 1} & \dots & \varphi_{m + X} + \gamma_{n + X}
    \end{tabular}
    \vspace{0.5em}
    \caption{General degree table for arbitrary $\varphi$ and $\gamma$.}
    \label{tab:degree_table}
\end{table*}

To be able to decode the product $AB$ from the coefficients of $h(x)$ we must have that the useful products $A_jB_{j'}$ are not disturbed by the interference terms such as $R_kB_{j'}$, $A_jS_{k'}$ or $R_kS_{k'}$. This is achieved by making sure that their corresponding exponents are distinct of the exponents of the useful terms.

If the matrices are partitioned using the inner product partitioning, then it is desirable to combine the subproducts $A_jB_{j'}$ as a sum when $j = j'$. On the other hand, if the matrices are partitioned with the outer product partitioning, then the subproducts $A_jB_{j'}$ need to be kept separate to be able to decode each of the terms individually.

The recovery threshold of such an SDMM scheme describes the smallest number of responses that are needed to interpolate the polynomial $h(x)$. This is obviously possible if the number of responses is at least $\deg(h(x)) + 1$, since a polynomial of degree $k$ can be interpolated from $k+1$ evaluations. However, by choosing the evaluation points carefully, it might be possible to interpolate the polynomial with fewer evaluations. The polynomial $h(x)$ will not necessarily have all terms, since the exponents $\varphi_j + \gamma_{j'}$ might have gaps in them. This means that the number of responses needed can be as low as the number of distinct elements in the degree table. This approach was taken in \cite{d2020gasp} and \cite{d2021degree} with the GASP codes.

Let us consider two examples of SDMM schemes that come from generalized Vandermonde matrices. Let us fix the parameters $p = m = n = 4$ and $X = 1$.

\newpage
\begin{example}[Secure MatDot]
The secure MatDot scheme was first introduced in \cite{aliasgari2019distributed}. The matrices are partitioned using the inner product partitioning and the exponents are chosen as
\begin{equation*}
    \varphi = (0, 1, 2, 3, 4), \quad \gamma = (3, 2, 1, 0, 4).
\end{equation*}
Then the degree table is the following with the useful terms highlighted.
\begin{center}
\begin{tabular}{c|cccccc}
      & 3 & 2 & 1 & 0 & 4 \\ \hline
    0 & \textbf{3} & 2 & 1 & 0 & 4 \\
    1 & 4 & \textbf{3} & 2 & 1 & 5 \\
    2 & 5 & 4 & \textbf{3} & 2 & 6 \\
    3 & 6 & 5 & 4 & \textbf{3} & 7 \\
    4 & 7 & 6 & 5 & 4 & 8
\end{tabular}
\end{center}

We can see that the first $p = 4$ diagonals are all equal to 3, which means that the coefficient of $x^3$ in $h(x)$ will contain the sum of the terms $A_jB_j$, which is exactly the product $AB$.
\end{example}

\begin{example}[GASP]
The GASP code was first introduced in \cite{d2020gasp} and was motivated by studying the combinatorics of the degree table. The matrices are partitioned using the outer product partitioning and the exponents are chosen as
\begin{equation*}
    \varphi = (0, 1, 2, 3, 16), \quad \gamma = (0, 4, 8, 12, 16).
\end{equation*}
Then the degree table is the following with the useful terms highlighted.
\begin{center}
\begin{tabular}{c|ccccc}
      & 0 & 4 & 8 & 12 & 16 \\ \hline
    0 & \textbf{0} & \textbf{4} & \textbf{8} & \textbf{12} & 16 \\
    1 & \textbf{1} & \textbf{5} & \textbf{9} & \textbf{13} & 17 \\
    2 & \textbf{2} & \textbf{6} & \textbf{10} & \textbf{14} & 18 \\
    3 & \textbf{3} & \textbf{7} & \textbf{11} & \textbf{15} & 19 \\
    16 & 16 & 20 & 24 & 28 & 32
\end{tabular}
\end{center}

The elements in the upper left corner are distinct, which means that the elements $A_jB_{j'}$ can be obtained as the coefficients of $\{1, x, \dots, x^{15} \}$ in $h(x)$.
\end{example}

\newpage
\section{Secure Distributed Gram Matrix Multiplication}

When applying the traditional SDMM methods to compute $AA^T$, we need to encode the matrix $A$ essentially twice. Once as $A$ and the second time as $A^T$. These encodings will be independent of each other, since they use independent random matrices for security and different linear codes. Furthermore, the computatation performed by the workers can not utilize any of the benefits of computing a Gram matrix, since the two encoded matrices are essentially unrelated.

In typical applications, the size of $AA^T$ is significantly smaller than the size of $A$. For example, in the least squares method $AA^T$ is $t \times t$ and $A$ is $t \times s$, where $t$ and $s$ are the number of variables and samples, respectively. The number of variables is usually significantly smaller than the number of samples, so $t^2 \ll ts$. Thus, we are interested in minimizing the upload cost of our scheme, while the download cost is negligible. We consider an SDMM scheme where both encodings are done using the same linear code, since this would only require one encoding of the matrix $A$, which will reduce the upload cost. The workers would then compute $\tilde{A}_i \tilde{A}_i^T$. Furthermore, it is enough that the workers compute and return the lower (or upper) triangular part, since the result will be symmetric. We call such a scheme a \emph{secure distributed Gram matrix multiplication (SDGMM)} scheme.

The security of an SDGMM scheme is analoguous to the SDMM schemes, \emph{i.e.}, an SDGMM scheme is secure against $X$-collusion if
\begin{equation*}
    I(\bm{A}; \tilde{\bm{A}}_\mathcal{X}) = 0
\end{equation*}
for all $\mathcal{X} \subset [N]$, $\lvert \mathcal{X} \rvert = X$. In this paper we will only be interested in the case when $X = 1$, which means that the scheme is secure, but does not provide protection against collusion. Similar schemes can be done for $X > 1$, but these require more elaborate constructions, since the security condition is more complicated. The following example will show how such a scheme could work.

\begin{example}
Partition the matrix $A$ horizontally to $p = 4$ equal sized pieces such that
\begin{equation*}
    A = \begin{pmatrix} A_1 & A_2 & A_3 & A_4 \end{pmatrix}.
\end{equation*}
Then $AA^T = \sum_{j=1}^4 A_j A_j^T$. Choose $R_1$ uniformly at random and define the polynomial
\begin{equation*}
    f(x) = A_1 + A_2x + A_3x^3 + A_4x^7 + R_1x^8.
\end{equation*}
Let $\alpha_1, \dots, \alpha_N \in \F_q^\times$ be distinct nonzero elements. Then each worker $i \in [N]$ receives the encoded matrix $\tilde{A}_i = f(\alpha_i)$ and computes $\tilde{A}_i \tilde{A}_i^T$. These responses are then evaluations of
\begin{align*}
    f(x) f(x)^T &= A_1 A_1^T + A_2 A_2^T x^2 \\
    &+ A_3 A_3^T x^6 + A_4 A_4^T x^{14} + (\text{other terms}).
\end{align*}
The degree of $f(x)$ is 16, so 17 responses are enough to interpolate the polynomial. Therefore, $N \geq 17$ is enough workers. The Gram matrix $AA^T$ can then be obtained by computing the sum of the coefficients of $1$, $x^2$, $x^6$ and~$x^{14}$.
\end{example}

If we were to compute the product $AA^T$ using the secure MatDot scheme with $p=4$ and $X=1$, we would require $N \geq 9$ workers. Thus, we would have $2 \cdot 9 = 18$ encoded matrices. Hence, the above example saves a small amount in the upload phase, since it only requires $17$ encoded matrices. On the other hand, if the evaluation points were chosen carefully, we would only need $N \geq 14$, since the polynomial $f(x)f(x)^T$ has just 14 nonzero terms.

\subsection{SDGMM Schemes Coming from a Degree Table}

Let us again consider the generator matrix as a generalized Vandermonde matrix with exponents $\varphi \in \N^{m + X}$. Without loss of generality, we assume that the exponents are ordered from smallest to largest. The first problem we encounter is that the degree table will be symmetric. This means that the requirements of an outer product partitioning scheme will be impossible to satisfy, since we can not distinguish the terms $A_j A_{j'}^T$ and $A_{j'} A_j^T$ for $j \neq j'$. On the other hand, the elements on the diagonal will all be distinct, since they are equal to $2\varphi_j$. This means that we can not recover the sum of the elements $A_jB_j$ as one of the coefficients of the resulting polynomial $h(x)$. However, we can still compute the sum as a linear combination of the responses.

The degree table associated with the above example is the following, where the useful terms have been highlighted.
\begin{center}
\begin{tabular}{c|ccccc}
      & 0 & 1 & 3 & 7 & 8 \\ \hline
    0 & \textbf{0} & 1 & 3 & 7 & 8 \\
    1 & 1 & \textbf{2} & 4 & 8 & 9 \\
    3 & 3 & 4 & \textbf{6} & 10 & 11 \\
    7 & 7 & 8 & 10 & \textbf{14} & 15 \\
    8 & 8 & 9 & 11 & 15 & 16
\end{tabular}
\end{center}

To construct a scheme for arbitrary $p$ and $X = 1$, we will choose the exponents $\varphi$ carefully to minimize the recovery threshold. We shall use the same partitioning as the above example, \emph{i.e.}, partitioning $A$ into $p$ pieces horizontally. We denote the elements in the degree table as
\begin{equation*}
    \mathcal{H} = \{ \varphi_i + \varphi_j \mid i, j \in [p + 1] \}.
\end{equation*}
We require that the diagonal elements $2\varphi_j$ for $j \in [p]$ are distinct from all other elements in the degree table. We say that $\varphi$ is \emph{valid} if it satisfies this property. This means that we can decode the products $A_j A_j^T$ from the responses. To do this we need at least $\lvert \mathcal{H} \rvert$ evaluations, since we have a system of linear equations with $\lvert \mathcal{H} \rvert$ unknowns. This gives us two options: minimize the largest value in the table (\emph{i.e.} $\varphi_{p+1}$), or minimize the number of distinct elements in the degree table (\emph{i.e.} $\lvert \mathcal{H} \rvert$).

\subsection{Minimizing the Number of Elements}

A trivial way to choose $\varphi \in \N^{p + X}$ is by setting $\varphi_1 = 0$ and $\varphi_{j+1} = 2\varphi_j + 1$. This will yield us with the largest element $\varphi_{p+X} = 2^{p+X} - 1$, which grows fast as $p$ grows. On the other hand, we notice that the number of elements in the degree table is $\frac{(p+X)(p+X+1)}{2}$. This means that the recovery threshold is essentially quadratic in $p$.

\emph{SDGMM scheme I:} Another way to choose $\varphi$ is by iteratively doubling the length. Let us assume that the target length $p + 1 = 2^n$. For $n = 0$, we can choose $\varphi^0 = (0)$. For $\varphi^{n+1}$ let us choose the concatenation of $\varphi^n$ and $\varphi^n + (2M + 1)\mathds{1}$, where $\mathds{1}$ is the all ones vector and $M$ is the largest element of $\varphi$. 

The degree table $\varphi^n \oplus \varphi^n$ will have largest element $2M$. Therefore,
\begin{equation*}
    \varphi^n \oplus (\varphi^n + (2M + 1)\mathds{1}) = \varphi^n \oplus \varphi^n + (2M + 1)\mathds{1},
\end{equation*}
where we also denote the all ones table with $\mathds{1}$, will have all elements at least $2M + 1$ and at most $2M + 1 + 2M = 4M + 1$. The elements in $(\varphi^n + (2M + 1)\mathds{1}) \oplus (\varphi^n + (2M + 1)\mathds{1}) = \varphi \oplus \varphi + (4M + 2)\mathds{1}$ are all at least $4M + 2$. Thus, the elements in the three boxes will be distinct. Additionally, the diagonal elements are distinct from all other elements. The total number of distinct elements is thus $R_{n+1} = 3R_n$.
\begin{center}
\begin{tabular}{C|C}
    \varphi^n \oplus \varphi^n & \varphi^n \oplus \varphi^n + (2M + 1)\mathds{1} \\ \hline
    \varphi^n \oplus \varphi^n + (2M + 1)\mathds{1} & \varphi^n \oplus \varphi^n + (4M + 2)\mathds{1}
\end{tabular}
\end{center}
As $R_0 = 1$, we get that $R_n = 3^n$. We had that $p + 1 = 2^n$, so $n = \ln(p+1) / \ln 2$. Hence,
\begin{align*}
    R = R_n &= 3^{\log_2(p+1)} = e^{\ln 3 / \ln 2 \ln(p+1)} = (p+1)^{\log_2 3}.
\end{align*}
When $p + 1$ is not a power of two, then we can round $p + 1$ up to the nearest power of two and take the first $p + 1$ exponents of the solution. This gives us that
\begin{equation*}
    R \leq (2(p+1))^{\log_2 3} = 3(p+1)^{\log_2 3}.
\end{equation*}
This shows us that the recovery threshold can be chosen to be $\Theta(p^{\log_2 3})$, which is better than quadratic as $\log_2 3 \approx 1.585$.

This construction is based on the $2 \times 2$ degree table of $\varphi = (0, 1)$. This means that a similar construction can be applied for powers of $k \in \{2, 3, \dots \}$ by considering an iterative definition based on a degree table of size $k \times k$. We see that $\varphi^0 \oplus \varphi^0$ contains all integers between $0$ and $2M$. From the above construction we see that if $\varphi^n \oplus \varphi^n$ contains all integers between $0$ and $M$, then $\varphi^{n+1} \oplus \varphi^{n+1}$ contains all integers between $0$ and $6M + 2$.

\subsection{Minimizing the Largest Element}

The reason to minimize the largest element is so that the decoding can be performed by regular polynomial interpolation of $h(x)$, which requires $\deg(h(x)) + 1$ evaluations. Then the evaluation points can be chosen freely, since the system will always be solvable. For this reason, it might be advantageous to consider this even if the degree table contains fewer terms than $\deg(h(x)) + 1$.

\emph{SDGMM scheme II:} Using a depth-first search algorithm we can search for an optimal solution $\varphi$ that produces a valid degree table with a minimal largest element. However, this does not give us a general solution. Below are some examples of the possible choices of $\varphi$ for different $p$ and $X = 1$. The general solution is not immediately clear from these examples. However, it turns out that the largest term grows quite fast with respect to $p$.
\begin{center}
\begin{tabular}{C|L}
    p & \varphi \\ \hline
    1 & (0, 1) \\
    2 & (0, 1, 3) \\
    3 & (0, 1, 3, 4) \\
    4 & (0, 1, 3, 7, 8) \\
    5 & (0, 1, 3, 4, 9, 10) \\
    6 & (0, 1, 3, 4, 9, 10, 12) \\
    7 & (0, 1, 3, 4, 9, 10, 12, 13) \\
    8 & (0, 1, 5, 6, 8, 13, 14, 17, 19) \\
    9 & (0, 1, 4, 6, 10, 15, 17, 18, 22, 23)
\end{tabular}
\end{center}

\subsection{Our Proposed Construction}

We propose the following construction for secure distributed Gram matrix multiplication over a finite field $\F_q$.

\begin{itemize}
    \item The input matrix $A \in \F_q^{t \times s}$ is paritioned to $p$ equal sized pieces horizontally.
    \item Given a valid list $\varphi \in \N^{p+1}$ of exponents, we define the polynomial
    \begin{equation*}
        f(x) = \sum_{j=1}^p A_j x^{\varphi_j} + R_1x^{\varphi_{p+1}},
    \end{equation*}
    where $R_1$ is a matrix of appropriate size with entries chosen uniformly at random from $\F_q$. Let $\alpha_1, \dots, \alpha_N \in \F_q^\times$ be distinct nonzero elements. Then worker $i \in [N]$ receives $\tilde{A}_i = f(\alpha_i)$ and computes the Gram matrix $\tilde{A}_i \tilde{A}_i^T$.
    \item The user interpolates the polynomial $f(x)f(x)^T$ from the responses and decodes the Gram matrix $AA^T$ by computing the sum of the coefficients of $x^{2\varphi_1}, \dots, x^{2\varphi_p}$.
\end{itemize}

We may choose this valid list of exponents $\varphi$ using either the SDGMM scheme I or II. The following theorem will show that the scheme is secure for $X = 1$. Larger values of $X$ are also possible with similar ideas, but are left for future research.

\begin{theorem}
Given a valid list of exponents $\varphi \in \N^{p+1}$, the associated SDGMM scheme is secure with $X = 1$. Furthermore, the scheme has recovery threshold $R = \lvert \mathcal{H} \rvert$, which is the number of elements in the degree table of $\varphi$. The upload and download costs of the scheme are $N ts / p$ and $Rt(t+1) / 2$.
\end{theorem}

\begin{proof}
By considering the generator matrix of the scheme, we notice that $F^{> p} = (\alpha_1^{\varphi_{p + 1}} \dots \alpha_N^{\varphi_{p + 1}})$. Hence, any $1 \times 1$ submatrix of $F^{>p}$ is invertible since $\alpha_i \neq 0$. By using the result from \cite[Theorem 1]{makkonen2022general}, the scheme is secure for $X = 1$.

By using a similar proof as \cite[Theorem 1]{d2020gasp}, we can choose the evaluation points $\alpha = (\alpha_1, \dots, \alpha_N) \in (\F_q^\times)^N$ from a large enough field such that the scheme is decodable for any subset of $R$ responses. The encoded pieces $\tilde{A}_i$ have size $t \times s / p$ and the responses have size $t \times t$. The responses are symmetric, so only the lower triangular part needs to be returned. The upload and download costs are then $N ts / p$ and $Rt(t+1) / 2$.
\end{proof}

The security of the scheme can also be seen by noticing that the encodings $\tilde{A}_i$ have the term $R_1\alpha_i^{\varphi_{p+1}}$, which is uniformly distributed as $\alpha_i \neq 0$. This means that the proposed scheme is secure, but does not provide any security against collusion.

\begin{remark}
We have a trivial upper bound on $\lvert \mathcal{H} \rvert$ as given by $R = \lvert \mathcal{H} \rvert \leq 2\varphi_{p+1} + 1$, since $\mathcal{H} \subseteq \{0, 1, \dots, 2\varphi_{p+1}\}$. In this case we may choose the evaluation points arbitrarily from $\F_q^\times$ as long as $q > N$. In this case the decoding can be performed by regular polynomial interpolation.
\end{remark}

\begin{remark}
Our scheme borrows some inspiration from the GASP scheme in \cite{d2020gasp, d2021degree} as it is obtained by optimizing the exponents in the degree table. However, we consider the inner product partitioning and a symmetric degree table, which has not been done previously in the literature to our knowledge.
\end{remark}

\subsection{Comparison}\label{sec:comparison}

\begin{figure}[t]
    \centering
    \includegraphics[width=0.7\columnwidth]{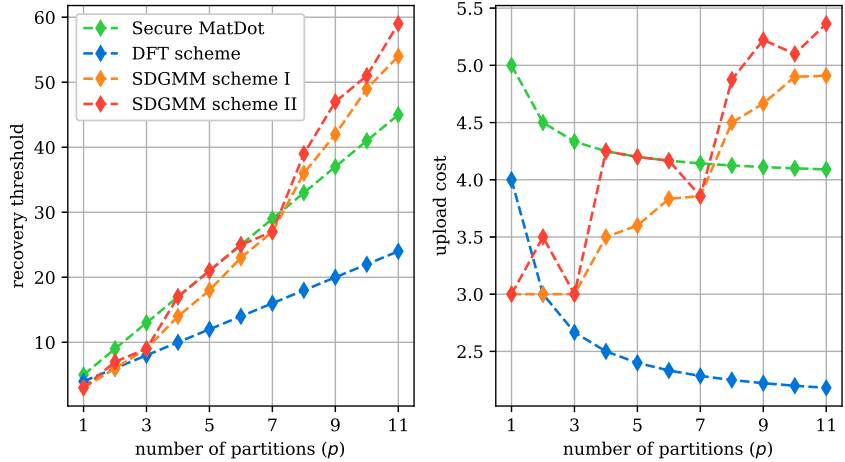}
    \caption{Comparision of the recovery threshold and upload cost of the proposed schemes, the secure MatDot scheme proposed in \cite{aliasgari2019distributed}, and the DFT scheme proposed in \cite{mital2022secure}. Our proposed schemes have lower upload cost compared to the secure MatDot scheme for small values of $p$. However, the DFT scheme has even lower upload cost, but is not able to tolerate stragglers.}
    \label{fig:plot}
\end{figure}

We wish to compare our proposed scheme to regular SDMM schemes in terms of the upload cost, due to the nature of our problem. It is natural to compare our scheme against the secure MatDot scheme, since it is able to achieve the lowest possible recovery threshold and upload cost according to the result in \cite{makkonen2022general}. We will also compare our scheme to the DFT scheme in \cite{mital2022secure}, which also uses the inner product partitioning. The GASP scheme is optimized to have a low download cost, while having a large upload cost, so we will not be comparing our proposed schemes to the GASP scheme. Gram matrix multiplication has also been studied in \cite{yu2019lagrange} with regard to linear regression using the Lagrange coded computation scheme. This scheme partitions the matrix similarly, but the workers need to compute significantly more compared to our proposed method. Hence, we will not be comparing our scheme to the one presented in \cite{yu2019lagrange}, since it would be difficult to fix the amount of computation that is done.

Figure \ref{fig:plot} shows the recovery threshold and upload cost as a function of $p$. We compare our schemes to the other schemes by fixing the amount of computation performed at the workers. As Gram matrix multiplication takes about half as many multiplications, we can consider the secure MatDot and the DFT schemes with $2p$ partitions and the SDGMM schemes with $p$ partitions. We see that our proposed schemes have a lower upload cost than the secure MatDot scheme for $p \leq 7$, which corresponds to $N \leq 27$ workers. For larger values of $p$ we see that our proposed schemes are worse than the secure MatDot scheme. The DFT scheme has a lower upload cost, but it cannot tolerate stragglers like the other schemes.

The schemes proposed in this paper consider the case of $X = 1$, \emph{i.e.}, the case where the workers do not collude. It is straightforward to extend these schemes to colluding with $X > 1$ workers, but it is not clear if the proposed schemes are optimal in this setting. We leave the case of $X > 1$ for future work.

\section{Analog Distributed Gram Matrix Multiplication}

The proposed scheme in the previous section works over a finite field $\F_q$, but for practical applications in data science, it would be preferable to have a scheme that works over the analog domain, \emph{i.e.}, over $\R$ or $\C$. The scheme can be adapted to work over the analog domain using similar methods that were done in \cite{makkonen2022analog} for regular SDMM schemes. These will give a tradeoff between the security and the numerical stability, but will achieve the same recovery thresholds and improvements over other schemes.

However, if security is not a concern, then the Gram matrix $AA^*$ can be computed by the following analog distributed Gram matrix multiplication (ADGMM) scheme. Here $A^*$ is the conjugate transpose of $A$. The matrix $A$ is partitioned to $p$ pieces horizontally. Define the polynomial
\begin{align*}
    f(x) = \sum_{j=1}^p A_j x^{j-1}.
\end{align*}
Let $\alpha_1, \dots, \alpha_N \in \C$ be distinct evaluation points such that $|\alpha_i| = 1$ for all $i \in [N]$. The worker $i$ receives $\tilde{A}_i = f(\alpha_i)$ and computes $\tilde{A}_i \tilde{A}_i^*$. These are evaluations of the function
\begin{align*}
    f(x)f(x)^* = \sum_{j=1}^p \sum_{j'=1}^p A_jA_{j'}^* x^{j-1} \overline{x}^{j-1} = \sum_{j=1}^p \sum_{j'=1}^p A_j A_{j'}^* x^{j - j'},
\end{align*}
since at the evaluation points $x = \alpha_i$, we have that $\overline{x} = \overline{\alpha}_i = \alpha_i^{-1}$. By multiplying this function with $x^{p-1}$, we get a polynomial of degree $2p - 2$, where the coefficient of $x^{p-1}$ is
\begin{align*}
    \sum_{j=1}^p A_j A_j^* = AA^*.
\end{align*}
This means that the Gram matrix $AA^*$ can be recovered from any $2p - 1$ responses. This scheme matches the MatDot scheme, but the matrix $A$ is encoded only once, which means that the upload cost is reduced to half compared to the regular MatDot scheme. Additionally, the workers can utilize more efficient algorithms to compute the Gram matrices $\tilde{A}_i \tilde{A}_i^*$.

We can also study the degree table for this scheme, where we consider the differences of the degrees instead of the sums, since $\overline{\alpha} = \alpha^{-1}$ for our chosen evaluation points. For example, for $p = 5$, we get the following degree table.
\begin{center}
\begin{tabular}{r|rrrrr}
      & 0 & 1 & 2 & 3 & 4 \\ \hline
    0 & \textbf{0} & -1 & -2 & -3 & -4 \\
    1 & 1 & \textbf{0} & -1 & -2 & -3 \\
    2 & 2 & 1 & \textbf{0} & -1 & -2 \\
    3 & 3 & 2 & 1 & \textbf{0} & -1 \\
    4 & 4 & 3 & 2 & 1 & \textbf{0}
\end{tabular}
\end{center}
Adding security to this scheme is not obvious, since any noise term $R_k x^{p+k - 1}$ would contribute an interference term $R_kR_k^*$ in the constant term, \emph{i.e.}, a zero on the diagonal by symmetry. In the secure MatDot scheme we are able to avoid this by not having a symmetric degree table.

\section*{Acknowledgements}

This work has been supported by the Academy of Finland under Grant No.\ 336005 and by the Vilho, Yrjö and Kalle Väisälä Foundation of the Finnish Academy of Science and Letters.

\bibliography{bib.bib}
\bibliographystyle{IEEEtran}

\end{document}